\documentclass[a4paper,11pt]{article}
\usepackage{amssymb, amsfonts, amsmath, amsthm}
\usepackage{geometry}
\usepackage[onehalfspacing]{setspace}
\usepackage{footnotebackref}
\usepackage[UKenglish]{babel}
\usepackage{endnotes}
\usepackage{indentfirst}
\usepackage[utf8]{inputenc}
\usepackage{natbib}
\usepackage{hyperref}
\hypersetup{
colorlinks=true,
linkcolor=Blue,
urlcolor=Blue,
citecolor=Blue,
bookmarksopen=false,
frenchlinks=false,
pdftitle={Competitive and Revenue-Optimal Pricing with Budgets},
}

\theoremstyle{plain}
\newtheorem{theorem}{Theorem}
\newtheorem{proposition}{Proposition}
\newtheorem{lemma}{Lemma}

\theoremstyle{definition}
\newtheorem{example}{Example}
\newtheorem{definition}{Definition}

\usepackage{subcaption}
\usepackage{mathtools, amssymb}
\usepackage{graphicx}
\usepackage{enumitem}
\usepackage{bm}
\usepackage[usenames,dvipsnames]{xcolor}
\usepackage[capitalize]{cleveref}
\usepackage{authblk}

\renewcommand{\vec}[1]{\bm{#1}}
\newcommand{\R}{\mathbb{R}}
\newcommand{\Rn}{\mathbb{R}^n}

\DeclareMathOperator{\conv}{\textnormal{conv}}
\DeclareMathOperator*{\argmax}{\textnormal{arg\,max}}
\newcommand{\bid}{\vec{b}}
\newcommand{\bids}{\mathcal{B}}
\newcommand{\eb}{\vec{e}}
\newcommand{\pb}{\vec{p}}
\newcommand{\qb}{\vec{q}}
\newcommand{\rb}{\vec{r}}
\newcommand{\xb}{\vec{x}}
\newcommand{\yb}{\vec{y}}
\newcommand{\zb}{\vec{z}}
\newcommand{\ssb}{\vec{s}}
\newcommand{\vb}{\vec{v}}
\newcommand{\nullgood}{money}

\usepackage{tikz}
\usetikzlibrary{shapes,patterns,calc,spy}
\usepackage{pgfplots}
\pgfplotsset{compat=1.17}
\usepackage{graphicx}

\usetikzlibrary{arrows,automata,math,backgrounds,decorations.pathreplacing,decorations.markings,patterns,arrows.meta,positioning}
\usepackage{tikz-network}

\definecolor{burgundy}{HTML}{882426}
\definecolor{midnight}{HTML}{151B54}
\definecolor{bluegray}{rgb}{0.4, 0.6, 0.8}
\definecolor{aqua}{rgb}{0.0, 1.0, 1.0}

\newcommand\bidsize{1.5mm}

\tikzset{price/.style={draw=black, circle, fill=black, very thick, minimum size=3.5, inner sep=0pt, outer sep=0pt}}
\tikzset{query/.style={draw=blue, fill=blue, circle, minimum size=0.4, inner sep=0pt, outer sep=0pt}}
\tikzset{querylabel/.style={font=\small}}
\tikzset{posbid/.style={draw=black, circle, thick, minimum size=\bidsize, inner sep=0pt, outer sep=0pt, fill=black}}
\tikzset{negbid/.style={draw=black, circle, very thick, minimum size=\bidsize, inner sep=0pt, outer sep=0pt, fill=white}}
\tikzset{boundingbox/.style={dashdotted, thick, color=blue}}
\tikzset{margin/.style={thick}}
\tikzset{hyperplane/.style={thick, dashed, color=gray}}
\tikzset{facet/.style={thick}}
\tikzset{bundle/.style={align=center, font=\scriptsize, color=blue}}
\tikzset{contour/.style={color=midnight}}
\tikzset{contourlabel/.style={draw=black, fill = white, rectangle, font=\tiny}}
\tikzset{boundary/.style={color=SeaGreen, ultra thick}}
\tikzset{closed/.style={draw=red, circle, fill=red, very thick, minimum size=1.5mm, inner sep=0pt, outer sep=0pt}}
\tikzset{open/.style={draw=red, circle, fill=white, very thick, minimum size=1.5mm, inner sep=0pt, outer sep=0pt}}

\tikzset{weight/.style={draw=gray, fill=white, circle, minimum size=0pt, inner sep=1pt, outer sep=0pt, font=\scriptsize}}

\tikzset{setI/.style={align=center, color=black}}

\newcommand\blfootnote[1]{
  \begingroup
  \renewcommand\thefootnote{}\footnote{#1}
  \addtocounter{footnote}{-1}
  \endgroup
}

\title{Competitive and Revenue-Optimal Pricing with Budgets}

\author[1]{Simon Finster}    
\author[2]{Paul W.~Goldberg}
\author[3]{Edwin Lock}
\affil[1]{CREST-ENSAE and Inria/FairPlay, \href{mailto:simon.finster@ensae.fr}{simon.finster@ensae.fr}}
\affil[2]{Department of Computer Science, University of Oxford, \href{mailto:paul.goldberg@cs.ox.ac.uk}{paul.goldberg@cs.ox.ac.uk}}
\affil[3]{Departments of Computer Science and Economics, University of Oxford, \href{mailto:edwin.lock@cs.ox.ac.uk}{edwin.lock@cs.ox.ac.uk}}

\date{8 April 2025}

\begin{document}

\maketitle

\begin{abstract}
In markets with budget-constrained buyers, competitive equilibria need not be efficient in the utilitarian sense, or maximise the seller's revenue. We consider a setting with multiple divisible goods. Competitive equilibrium outcomes, and only those, are constrained utilitarian efficient, a notion of utilitarian efficiency that respects buyers' demands and budgets. Our main contribution establishes that, when buyers have linear valuations, competitive equilibrium prices are unique and revenue-optimal for a zero-cost seller.

\vspace{0.5em}

\noindent \textbf{Keywords:} competitive equilibrium, revenue maximisation, efficiency, market design, budget constraints, Fisher markets, product-mix auctions, arctic auction

\vspace{0.5em}

\noindent \textbf{JEL codes:}
D40,  
D44,  
D50  
\end{abstract}

\blfootnote{\emph{Acknowledgments}: We thank Péter Eső, Bernhard Kasberger, Paul Klemperer, Alexander Teytelboym, Zaifu Yang, and anonymous reviewers for insightful discussions and comments. During the work on the final version of the paper, Goldberg and Lock were supported by a JP Morgan faculty fellowship and EPSRC grant EP/X040461/1. Work on the final version of this project was also supported by the National Science Foundation under Grant No. DMS-1928930 and by the Alfred P. Sloan Foundation under grant G-2021-16778, while Finster was in residence at the Simons Laufer Mathematical Sciences Institute (formerly MSRI) in Berkeley, California, during the Fall 2023 semester.

A preliminary version of this article was presented at WINE'23 under the title ``Substitutes markets with budget constraints: solving for competitive and optimal prices''.}

\section{Introduction}
\label{sec:introduction}


The standard model of a seller is that of revenue maximisation. But client satisfaction and participation may lead sellers to also prioritise social welfare and consider competitive pricing. In many settings, these are conflicting objectives; e.g., when agents have concave values, competitive pricing leads to lower revenue for the seller than revenue-optimal pricing (as we demonstrate in \cref{sec:examples}). In markets with budget constraints there is an additional tradeoff: It is well-known in classic economic theory that competitive equilibria need not achieve utilitarian efficiency. 
Still, it is natural to consider a notion of the best ``achievable'' social welfare in the presence of budgets.
For instance, \citet{Che-2013} develop a mechanism for a simple market%
\footnote{\citet{Che-2013} characterize the socially optimal mechanism that achieves an efficient assignment of a homogeneous supply of a single good to a finite mass of budget-constrained agents.}
that maximizes expected social welfare, but respects the agents' (private) budget constraints, using a random assignment of supply and/or cash subsidies. Similarly, seller-optimal mechanisms are described in \citet{Che-2000} and \citet{Pai-2014} with implementations as non-linear pricing schemes and modified all-pay auctions, respectively.

Two important applications of our model are the digital advertising economy and markets for the exchange of financial assets. In these multi-product environments, the mechanisms described above are not immediately applicable. Moreover, given the fast-paced and large-scale nature of online ad auctions, which occur within fractions of a second, their implementation appears to be too complex. In exchanges for financial assets (\citet{Klemperer2018}, see also \cref{sec:arctic}), the seller may have limited control over resale, so randomized allocations and cash transfers may attract speculators. Instead, we consider a seller who is restricted to a price-only mechanism. 

In our market, the seller supplies multiple divisible goods in finite quantities at zero cost to multiple buyers. Each buyer has monotone and concave, possibly linear, valuations and quasi-linear utility, and is endowed with a finite budget of money. In our price-only mechanism, the seller chooses uniform prices that, anticipating the buyers' demand correspondences, respect budgets and market supply. The seller may choose competitive prices that clear the market, or she may retain some of her supply and raise prices to maximise revenue.

We consider the notion of \emph{constrained social welfare}, i.e.,~social welfare among all allocations and prices%
\footnote{Any feasible allocation is inherently tied to prices in our market, as feasibility must respect the buyers' budgets.} that respect the buyers' demands, budgets, and market supply.
When buyers have quasi-linear preferences, constrained social welfare is maximised in a competitive equilibrium, and in a competitive equilibrium only (\cref{proposition:welfare}).%
\footnote{This theorem continues to hold if the seller has quasi-linear utility with non-zero costs.}
We also say that a competitive equilibrium is \emph{constrained utilitarian efficient}, or simply \textit{constrained efficient}. As is standard, a competitive equilibrium is Pareto efficient, but not necessarily utilitarian efficient (see also \cref{sec:examples}).

Our main result establishes that when buyers have linear valuations, the unique competitive equilibrium also maximises the seller's revenue in a price-only mechanism (\cref{theorem:coincidence}). The unique market-clearing prices are buyer-optimal among all revenue-maximising prices, as they maximise the quantities allocated to each buyer. However, there exist revenue-maximising outcomes that are not the competitive equilibrium, and therefore are not constrained utilitarian efficient. We prove our result by studying the \emph{feasible region}, defined as the set of prices at which, for every good, either the market clears (respecting buyers' budgets and demands) or there is excess supply. We show that this non-convex region has elementwise-smallest prices, and that these prices clear the market as well as maximise the seller's revenue.

The notion of constrained utilitarian efficiency provides a compelling and practical benchmark for social optimality in markets with budgets, where competitive equilibrium allocations need not be classically utilitarian efficient. While competitive equilibria in our market remain Pareto efficient, welfare considerations require the stronger criterion of utilitarian efficiency. Studying the constrained social optimum is especially important in settings in which it would be hard to accept breaking or circumventing budget constraints for the sake of efficiency. This is the case in our two motivating examples, digital advertising, and financial asset exchanges. In an ad auction, although a small business may derive a high value from an ad placement, only their \emph{ability to pay} can be relevant to a for-profit digital platform. Similarly, in the exchange of financial assets, the budget of a buyer represents the limit on the nominal value of their asset that is to be exchanged, and thus the seller would not want to grant an allocation of substitute assets exceeding this limit.\footnote{In general, the analysis of social welfare that transcends buyers' budget constraints may be more relevant when these budget constraints are linked to individuals with unequal endowments or socially indispensable businesses with unequal access to capital markets.}


Our market has been studied as a simple model of many ad auctions as they occur on digital market platforms \citep{Conitzer-2022}, and is also called a mixed Fisher or quasi-Fisher market \citep{Chen2007,Murray2020}. When businesses compete for digital advertising space, the decision of which publisher (product) to bid for is non-trivial. It is intuitive to choose an advertising budget and state demand in terms of ``limit market prices'' for multiple, distinct products. The seller or platform assigns to each buyer those products that yield the highest value for money to them. Our main result suggests that the platform can set prices that are both revenue maximizing and socially optimal within the feasibility constraints of advertiser budgets. However, it does not imply that the seller always sets socially optimal prices.

Budget-constrained buyers also appear in exchanges for financial assets. For example, \citet{Klemperer2018} introduces the ``arctic auction'', originally developed for the government of Iceland, who planned to use this auction to exchange blocked accounts for other financial assets such as cash or bonds. Buyers could submit a budget and their trade-offs between different assets, and the auction was solved to maximise the seller's revenue. Quasi-Fisher markets can be interpreted as a special case of this auction, as we discuss in \cref{sec:arctic}. Further applications include debt restructuring and the (re-)division of firms between shareholders (see also \citet{Klemperer2018,Baldwin2024language}).

\subsection{Related Literature}
\label{sec:literature}

The practical relevance of our market is highlighted in \citet{Conitzer-2022}, whose setting is particularly inspired by online ad auctions. The basic properties of the market and budget-constrained buyers are identical to ours, although they consider only linear valuations. Contrasting our setting, each divisible good is sold in an independent, single-unit first-price auction, in which only the highest bidders can win a positive quantity. \citet{Conitzer-2022} introduce the solution concept of \emph{first price pacing equilibria} (FPPE), in which the submitted bids correspond to the buyers' values scaled (uniformly for each buyer) by a pacing multiplier.\footnote{An FPPE is defined as a set of pacing multipliers (one for each buyer) and allocations that satisfy the allocation and pricing rule of standard first-price auctions, as well as budget feasibility, supply feasibility, market clearing for demanded goods, and `no unnecessary pacing', i.e.~a buyer's multiplier equals one if she has unspent budget.}
Interestingly, this at first sight unrelated auction procedure can also be solved using the modified Eisenberg-Gale convex programme of \citet{Chen2007} that was proposed for mixed Fisher markets. Moreover, \citet{Conitzer-2022} show that the unique FPPE corresponds to a competitive equilibrium in the sense of our setting; that is, in the overarching market for all goods with budget-constrained, quasi-linear buyers. While they show that the FPPE is revenue-maximal among all budget-feasible pacing multipliers and corresponding allocations, our work implies that the FPPE is indeed revenue-maximising in the entire market.

Several papers have studied optimal mechanism and auction design in the presence of budget constraints. A crucial distinction between those mechanisms and ours is the focus on incentives that arise from the presence and extent of private information. For example, \cite{Laffont-1996} characterize the optimal auction, and \cite{Maskin-2000} designed the constrained-efficient mechanism when budgets are known. The optimal and constrained-efficient mechanism in different settings with private budgets and values, including a single buyer, multiple buyers, or a population of buyers, are developed, for example, in \cite{Che-2000,Che-2013,Pai-2014,Richter-2019}. In contrast to this literature, our approach is more practical, focusing on a market with a supply of differentiated goods in which the seller is restricted to a price-only mechanism (see also the discussion in \cref{sec:introduction}), and buyers behave non-strategically.

The market with budget-constrained buyers with quasi-linear utilities and linear valuations is also known as a quasi-Fisher market \citep{Murray2020}, as it can be considered a generalisation of standard Fisher markets \citep{Brainard-Scarf2005}.
Quasi-Fisher markets have appeared in various guises, mainly in a computational context. The first results on quasi-linear Fisher markets were developed by \citet{Chen2007}, who showed that competitive equilibria can be computed in polynomial time, with several others to follow.%
\footnote{A more comprehensive overview of computational contributions on linear and quasi-linear Fisher markets is given in our preliminary working paper \citep{finster2023substitutes}.}
In standard, linear Fisher markets, buyers spend their entire budget at any market prices, and so revenue is constant at all prices. In contrast, buyers with quasi-linear utility and budget constraints spend nothing when prices are unacceptably high. Hence, the notion of maximising revenue becomes a viable objective for the seller to pursue. Contrasting the literature on quasi-Fisher markets, our paper considers not only competitive equilibrium, but also maximising revenue. Our unifying result on constrained utilitarian efficiency, competitive equilibrium, and revenue demonstrates the importance of the quasi-linear setting from a theoretical perspective as well as in applications.

Our market is also equivalent to a version of the 
Arctic Product-Mix Auction (PMA) \citep{Klemperer2018}, in which the seller's costs are zero. The Arctic PMA was originally designed for the government of Iceland in order to allow to exchange offshore accounts for other financial assets. \citet{fichtl2022computing} studies this auction from the perspective of computing revenue-maximising prices.

\paragraph*{Organisation.} The remainder of the paper is structured as follows. In \cref{sec:examples}, we motivate the research with several examples in which buyers' values are not linear, and the coincidence of revenue-optimality and constrained utilitarian efficiency fails. \cref{sec:markets-preferences-objectives} describes the model, and \cref{section:market-outcomes} defines market outcomes and their properties. In \cref{sec:constrained-efficiency}, we prove the constrained efficiency of competitive equilibrium, and in \cref{section:revenue-welfare-coincidence}, we show the coincidence of constrained efficiency and revenue-optimality in the linear case. \cref{sec:arctic} establishes the connection between our market and the Arctic auction, and \cref{sec:conclusion} concludes.

\paragraph*{Notation.}
For any two vectors $\vec{v}, \vec{w} \in \mathbb{R}^n$, we write $\vec{v} \cdot \vec{w}$ for their dot product, and ${\vec{v} \leq \vec{w}}$ when the inequality holds elementwise. For any $j \in \{1, \ldots, n\}$, $\eb^j$ denotes the $n$-dimensional indicator vector with $e^j_{j} = 1$ and $e^j_k = 0$ for all $k \neq j$. We also define $\eb^0 \coloneqq \bm{0}$.

\section{Examples: Efficiency, Constrained Efficiency, and Revenue}
\label{sec:examples}

We illustrate the objectives of maximizing social welfare, constrained social welfare, and the seller's revenue. \Cref{example:no-efficiency} shows that competitive equilibrium with budget constraints, which is Pareto efficient, is also constrained utilitarian efficient, but not necessarily fully utilitarian efficient. In \cref{example:concave-values-1}, we show that with diminishing marginal values, the seller's revenue is not maximised in a constrained-efficient allocation. \Cref{example:two} illustrates our main result (\cref{theorem:coincidence}), the coincidence of competitive equilibrium and revenue-optimality, in a setting with two goods and constant marginal values. Note that allocations are non-negative throughout the paper.
\begin{example}
\label{example:no-efficiency}
Consider a market with two buyers and a seller with one good in unit supply and zero costs. The buyers' per-unit values are $v^1 > v^2$ and their budgets are $\beta^1 > \beta^2$, and we assume $\beta^1 + \beta^2 \leq v^2$. The fully efficient (social welfare maximizing) allocation gives the entire unit to the first buyer.
However, due to their budgets, for any price $p\leq v^2$, each buyer $i$ demands quantity $\frac{b^i}{p}$. To clear the market, the price must satisfy $\frac{\beta^1}{p} + \frac{\beta^2}{p} = 1$, i.e.,~$p = \beta^1 + \beta^2$. The allocation to buyer 1 is therefore $\frac{\beta^1}{\beta^1 + \beta^2} < 1$.
\end{example}
The competitive equilibrium in \cref{example:no-efficiency} is not fully efficient, but it maximises constrained social welfare that respects the buyers' budgets. Budgets intrinsically link constrained welfare considerations to prices. We say that an outcome, consisting of an allocation and prices, is feasible if it respects the buyers' demands, budgets, and the supply constraint. The constrained-efficient outcome is defined as maximizing social welfare among all feasible outcomes. In the example, social welfare is maximised by assigning as much as possible to buyer 1. However, if we cannot violate buyers' budgets, increasing the quantity allocated to the first buyer beyond $\frac{\beta^1}{\beta^1 + \beta^2}$ requires decreasing the price. This would again increase the demand of the second buyer and violate the supply constraint, tipping the market out of equilibrium. Thus, the competitive equilibrium is constrained utilitarian efficient. The seller's revenue in the competitive equilibrium is $\beta^1 + \beta^2$. This is clearly optimal, as the seller cannot hope to obtain more money than is available in the market.

Although competitive equilibrium is always constrained efficient, with any number of goods, buyers, and irrespective of the valuation, the coincidence of revenue-optimality and constrained efficiency is more fragile. The following example shows that, with diminishing marginal values, it may fail.

\begin{example}
\label{example:concave-values-1}
The seller provides a single good with supply $s\in \R_+$ at zero costs. There is one buyer with a continuous valuation $v:\R\to\R$ for the good and quasi-linear utility $u(p,x) = v(x) - px$, where $p$ denotes the unit price of the good, and a budget $\beta\in \R_+\cup\{\infty\}$.  We aim to find outcomes $(p,x)$ that maximise either constrained efficiency or the seller's revenue in a price-only mechanism, respectively.
\end{example}

Maximising revenue with a price-only mechanism, the seller sets a price in anticipation of the buyer's demand. The buyer demands the quantity $x$ at price $p$ that maximises $v(x) - px$ subject to not exceeding the budget, $px \leq \beta$. Thus, the seller's revenue is maximised at price $p$ and allocation $x \leq s$ that maximise $px$ and the buyer demands $x$ at $p$.
Social welfare, on the other hand, is highest at price $p$ and an allocation $x \leq s$ that maximises $v(x)$ among all $(p, x)$ for which the buyer demands $x$ at $p$.

To derive allocations and prices, we first have to make some assumptions on the buyer's valuation. A typical assumption is diminishing marginal values. In that case, the buyer always demands a quantity such that their marginal utility at this quantity is zero. Assuming $v'(s)\cdot s \leq \beta$, social welfare is maximised if $v'(x) = p$ and $x = s$, and markets clear. However, if the seller were allowed to adjust the price, anticipating the buyer's demand, she might be able to extract more revenue (and not sell the entire supply). The following proposition demonstrates that this is indeed the case for all strongly concave value functions if the buyer's budget is large.
\begin{proposition}
\label{prop:concave-v}
    Let $v$ be differentiable and strongly concave with parameter $m$ for some $m>0$, and let $\beta = \infty$. Then there exists some supply $s \in \R$ so that revenue is not maximised at the market-clearing price.
\end{proposition}
When the buyer has a finite budget, we assume that supply lies in the finite interval $X:=[0, \max\{x \mid x v'(x) \leq \beta \}]$. We define $\Tilde{x}$ implicitly by $\Tilde{x}v'(\Tilde{x}) = \beta$.

\begin{proposition}
\label{prop:concave-v-2}
Suppose the buyer has a strongly concave valuation $v$ with parameter $m$ and a finite budget~$\beta$. If supply $s$ is contained in $X$ with $v'(s) < ms$, or if $m > \frac{v'(\Tilde{x})}{\tilde{x}}$, then revenue is not maximised at the market-clearing price.
\end{proposition}

In other words, for any strongly concave valuation, we can find a combination of budget and supply such that the maximisers of the constrained social welfare and the revenue maximisation problem do not coincide. For example, we may require the valuation to be sufficiently concave relative to the budget.
The arguments of the above propositions are standard in monopoly theory, but for completeness we provide the proofs in the appendix. \Cref{example:concave-no-budget} in the appendix further illustrates the case with large budgets.


In light of the above propositions, we consider the class of constant marginal values, i.e., $v(x) = vx$ for some scalar per-unit valuation $v$.
\begin{proposition}\label{prop:concave-v-3}
Suppose the buyer has valuation $v(x) = vx$ and budget $\beta$. Then the seller's revenue is maximised at market-clearing prices.
\end{proposition}
The above proposition is straightforward. Constrained social welfare is maximised at $p = \min\{v, \frac{\beta}{s}\}$ and $x = s$. We argue that the constrained efficient allocation and price are also revenue-optimal, because the seller extracts the maximum possible revenue at $\pb$. Suppose first that $p = v \leq \frac{\beta}{s}$. At prices strictly above $v$, the buyer demands nothing, so the seller gets a revenue of zero. Prices strictly below $v$ are not feasible, as the buyer demands more than one unit of the good. In the case that $p = \frac{\beta}{s} < v$, the buyer exhausts their entire budget demanding $x = s$, so the seller cannot extract more revenue.

An immediate question to ask is whether this reasoning extends to more general environments and preferences. The answer we present in this paper is affirmative: if any number of buyers have quasi-linear, budget-constrained utility and linear values for any number of goods, then seller-optimal revenue and constrained efficiency are attained at the unique set of elementwise-minimal prices. This result, however, is not immediate. In the following, we illustrate the difficulty in another simple example with two goods.

\begin{figure}[htp]
    \centering
    \begin{tikzpicture}[scale=0.8]
        \begin{axis}[
        axis lines=middle,
        axis line style={->},
        x label style={anchor=west},
        y label style={anchor=south},
        xlabel={$p_A$}, ylabel={$p_B$},
        xtick = {1,2,3,4},
        ytick = {1,2,3},
        xmin=0, xmax=5,
        ymin=0, ymax=3.5,
        grid=major,
        grid style=dashed,
        clip=false,
        ]
        \fill[gray,opacity=0.2](5,3.5)--(5,1.5)--(3,1.5)--(2,1)--(1,1)--(0.667,0.667)--(0.667,1)--(1,1.5)--(1,3.5)--cycle;
        \node[posbid, label=25:$\vb^1$] at (2,3) {};
        \node[posbid, label=25:$\vb^2$] at (2,2) {};
        \node[posbid, label=25:$\vb^3$] at (4,2) {};
        \draw[margin](0,0) -- (2,3);
        \draw[margin](2,3)--(5,3);
        \draw[margin](0,0)--(2,2)--(2,3.5);
        \draw[margin](2,2)--(5,2);
        \draw[margin](0,0)--(4,2)--(4,3.5);\draw(4,2)--(5,2);
        \draw[boundary](1,3.5)--(1,1.5)--(0.667,1)--(0.667,0.667)--(0.6,0.6);
        \draw[boundary](0.667,0.667)--(1,1)--(2,1)--(3,1.5)--(5,1.5);
        \node[red] at(0.6,0.6){\small$\bullet$};
        \node[blue] at(0.667,0.667){\small$\bullet$};
        \end{axis}
    \end{tikzpicture}

    \caption{The feasible region in price space corresponding to \cref{example:two}.}
    \label{fig:example-two}
\end{figure}

\begin{example}
\label{example:two}
Two goods, $A$ and $B$, are for sale with a supply of $s_A = 3$ and $s_B = 2$. There are three buyers $1,2$, and $3$ with the following marginal values: $\vb^1 = (v^1_A, v^1_B) = (2,3)$, $\vb^2 =(v^2_A, v^2_B) = (2,2)$, and $\vb^3 =(v^3_A, v^3_B) = (4,2)$. Utilities are quasi-linear, i.e.~$u^i(x,p) = \sum_j (v^i_j - p_j)x_j$ for buyer $i$. Each buyer has a budget of $\beta^1 = \beta^2 = \beta^3 = 1$.
\end{example}

In this example market, we can allocate the three units of $A$ and 2 units of $B$ among the three buyers to maximise either revenue or social welfare, but need to respect individual demand. It is not hard to check that at given prices $(p_A,p_B)$ each buyer $i$ will demand a good $j \in \argmax_{j=1,2}\frac{v^i_j}{p_j}$ if $v^i_j \geq p_j$. Any good $k$ with $v^i_k < p_k$ will never be demanded by buyer $i$. This kind of individual demand can be easily represented in price space (more detail in \cref{sec:arctic}). At some prices, aggregate demand is too large to be satisfied by supply. Prices at which aggregate demand does \emph{not} exceed supply are called \emph{feasible}. The set of feasible prices makes up the \emph{feasible region}. The bids (black dots) and the feasible region (in grey) for \cref{example:two} are illustrated in \cref{fig:example-two}. Note that the feasible region also includes a short line segment between $\pb^*:=(\frac{3}{5},\frac{3}{5})$ (red dot) and $\pb':=(\frac{2}{3},\frac{2}{3})$ (blue dot). The feasible region has the key property that for any pair of feasible price vectors, their elementwise minimum also belongs to the region (\cref{proposition:elementwise-minimal-prices}).

At prices $(\frac{3}{5},\frac{3}{5})$, buyer 1 demands $\frac{5}{3}$ of $B$, buyer 3 demands $\frac{5}{3}$ of $A$, and buyer 2 demands $ x^2_A\in[0,\frac{5}{3}]$ copies of $A$ and $\frac{5}{3}-x^2_A$ of $B$. With supply $(s_A,s_B) = (3,2)$, we set $x_A^2=\frac{4}{3}$ to clear the market. It is easy to check that any prices on the line segment $[p^*,p']$ induce a feasible allocation. All prices $[p^*,p']$ are revenue-maximising. However, only $\pb^*$ clears the market.

\section{The Market}
\label{sec:markets-preferences-objectives}
We have $m$ buyers $[m] \coloneqq \{1, \ldots, m \}$, $n$ divisible goods $[n] := \{1, \ldots, n\}$, and a (divisible) numeraire, denoted $0$, that we call ``money''. Let $[n]_0 = \{0, \ldots, n\}$ be the goods together with money. A \emph{bundle}, typically denoted $\xb$ or $\yb$, is an $n$-dimensional vector of non-negative reals whose entry $x_j \geq 0$ for each $j \in [n]$ denotes the {quantity} of good~$j$.
The seller (``she'') is endowed with a \textit{supply bundle} $\vec{s} \in \Rn_+$ that she wishes to sell, partially or completely, by setting uniform, non-negative market prices $\pb \in \Rn_+$ for the goods. She has zero costs and no utility for leftover supply,%
\footnote{This assumption is valid, e.g., in ad auctions. Unsold advertising slots generally have no utility to the seller (unless it was using the slots for advertising their own products). Some of our results extend to non-zero seller costs.}
so her quasi-linear utility for selling bundle $\xb$ at prices $\pb$ is given by $\pb \cdot \xb$.

Each buyer $i \in [m]$ (``he'') has a \emph{budget} $\beta^i$ and a \emph{valuation} $v^i$ mapping every bundle $\xb$ to a value $v^i(\xb) \in \R_+$. We assume that $v^i$ is concave and monotone increasing.
Buyers have \textit{quasi-linear utility} $u^{i}(\xb, \pb) \coloneqq v^i(\xb) + \beta^i - \pb \cdot \xb$ from receiving bundle $\xb$ at prices $\pb$, where $\beta^i - \pb \cdot \xb$ is their leftover money. A buyer's \textit{demand} at prices $\pb$ consists of the bundles $\xb$ that maximise his utility $u^{i}(\cdot, \pb)$, subject to not exceeding his budget. The latter is expressed by the budget constraint $\pb \cdot \xb \leq \beta^{i}$. This leads to the budget-constrained \textit{demand correspondence}
\[
    D^i(\pb) \coloneqq \argmax_{\xb \in \Rn_+, \pb \cdot \xb \leq \beta^{i}} (v^i(\xb) - \pb \cdot \xb),
\]
omitting the constant $\beta^i$.
In \cref{section:revenue-welfare-coincidence,sec:arctic}, we consider the setting in which buyers have linear valuations $v^i(\xb) = \vb^i \cdot \xb$ given by a vector $\vb^i \in \Rn_+$. This market is also called a mixed Fisher or quasi-Fisher market \citep{Chen2007,Murray2020}.%
\footnote{Quasi-Fisher markets differ from classical linear Fisher markets in that buyers in the latter market have linear (as opposed to quasilinear) utilities $u^i(\xb) = \vb^i \cdot \xb$. So their demand correspondences are given by $D^i(\pb) = \argmax_{\xb \in \R_+^n, \pb \cdot \xb \leq \beta^i} \vb^i \cdot \xb$.}

\subsection{Market Outcomes}
\label{section:market-outcomes}
The seller solves the auction by determining a market \textit{outcome} $(\pb, (\xb^i)_{i \in [m]})$, which consists of \textit{market prices} $\pb$ and an \textit{allocation bundle} $\xb^i$ to each buyer $i \in [m]$. For brevity, we drop the subscript $i \in [m]$ when denoting an allocation $(\xb^i)$ to buyers, as buyers are fixed throughout.
When determining a market outcome, the seller wishes not to exceed supply. Allocations are also restricted to bundles that reflect buyers' demands and budgets at the chosen market prices.
In particular, each buyer $i$ receives a bundle $\xb^i$ that maximises his utility $u^i(\yb, \pb)$ among all bundles $\yb$ that do not cause his to spend an amount $\pb \cdot \yb$ exceeding his budget $\beta^i$.
We call such outcomes \textit{feasible}.

\begin{definition}
\label{definition:feasible-outcome}
A market outcome $(\pb, (\xb^i))$ is \textit{feasible} if:
\begin{enumerate}[label=(\roman*)]
\item the aggregate allocation of each good $j$ does not exceed its supply, so $\sum_{i \in [m]} \xb^i \leq \ssb$;
\item each buyer demands his allocation at $\pb$, so $\xb^i \in D^{i}(\pb)$ for all $i \in [m]$.
\end{enumerate}
\end{definition}

Prices are \textit{feasible} if they can be extended to a feasible outcome with some allocation. In \cref{section:revenue-welfare-coincidence}, we will study the geometry of the feasible region, which consists of the set of all feasible prices. Conversely, we say that an allocation $(\xb^i)$ is \textit{supported} by some prices $\pb$ if $(\pb, (\xb^i))$ form a feasible outcome. This means that the allocation can be implemented by some prices~$\pb$ at which each buyer $i$ demands bundle $\xb^i$.

\begin{definition}
\label{definition:feasible-region}
We say that prices $\pb$ are \emph{feasible} if there exists an allocation $(\xb^i)$ so that $(\pb, (\xb^i))$ is a feasible outcome. The \emph{feasible region} consists of all feasible prices.
\end{definition}

The seller's \textit{revenue} from an outcome $(\pb, (\xb^i))$ is given by $\sum_{i \in [m]} \pb \cdot \xb^i$. In order to maximise revenue, the seller may set prices at which her supply is not cleared.

However, long-term considerations such as client satisfaction and participation may also lead the seller to consider competitive pricing. A \textit{competitive equilibrium} consist of a feasible outcome that clears the market. The existence of competitive equilibrium is guaranteed in our market, as it can be seen as an Arrow-Debreu exchange economy \citep{Arrow-Debreu1954,Chen2007}. 

\begin{definition}
A market outcome $(\pb, (\xb^i))$ is a \textit{competitive equilibrium} if it is feasible and clears the market for all positively-priced goods, so $\sum_{i \in [m]} x^i_j = s_j$ for all goods $j$ with $p_j>0$.
\end{definition}

\section{Constrained Efficiency}
\label{sec:constrained-efficiency}
The \textit{social welfare} of an allocation $(\xb^i)$ is $\sum_{i \in [m]} v^i(\xb^i)$, and the social welfare of an outcome is the welfare of its allocation. In markets with unlimited spending power, competitive equilibrium allocations are Pareto efficient and utilitarian efficient, i.e.,~maximise social welfare. In markets with budgets, however, competitive equilibrium allocations need not be utilitarian efficient, as welfare-maximising allocations may not be supported by any market prices. \Cref{example:no-efficiency} in \cref{sec:examples} illustrates.

Instead, we consider the maximum social welfare attainable by an allocation that is supported by some market prices, i.e., the maximum social welfare achievable by a market outcome. We call a feasible outcome \textit{constrained (utilitarian) efficient} if it maximises social welfare among all feasible outcomes.%
\footnote{The term `constrained efficiency' is motivated by the fact that the allocations we consider must be supported by a price to form a feasible outcome (cf.~\cref{definition:feasible-outcome}), and feasibility is constrained by budgets.}
Our notion of constrained efficiency can thus be considered a variation of efficiency in the presence of budgets.

\begin{definition}
\label{definition:constrained-efficiency}
A market outcome $(\pb, (\xb^i))$ is \textit{constrained (utilitarian) efficient} if it is feasible and maximises social welfare, $\sum_{i \in [m]} v^i(\xb^i)$, among all feasible outcomes.
\end{definition}

First, we establish that competitive equilibrium maximises constrained efficiency, and any feasible outcome maximising constrained efficiency constitutes a competitive equilibrium.\footnote{This is reminiscent of the fundamental welfare theorems, but we consider utilitarian efficiency instead of Pareto efficiency. The standard welfare theorems are known to hold in the Arrow-Debreu exchange economy, which embeds our market.}
This classic result is well-known for markets with indivisible goods without budgets (see, e.g.,~\cite{Sun2014}) and we extend this to our market in which outcomes must respect budgets, demand and supply.

\begin{proposition}
\label{proposition:welfare}
A market outcome is a competitive equilibrium iff it is constrained efficient.
\end{proposition}
\begin{proof}
Fix a competitive equilibrium $(\pb, (\xb^i))$ and let $(\qb, (\yb^i))$ be any \emph{feasible} outcome.
For convenience, we define the aggregately allocated bundles $\xb = \sum_{i \in [m]} \xb^i$ and $\yb = \sum_{i \in [m]}\yb^i$.
As $(\pb, (\xb^i))$ is feasible, the quasi-linearity of utilities tells us that
$v^i(\xb^i) + \beta^i - \pb \cdot \xb^i \geq v^i(\yb^i) + \beta^i - \pb \cdot \yb^i$ for every buyer $i$, so
\begin{equation}
\label{eq:welfare-theorem}
\sum_{i \in [m]} v^i(\xb^i) - \sum_{i \in [m]} v^i(\yb^i) \geq \sum_{i \in [m]} \pb \cdot (\xb^i - \yb^i) = \pb \cdot (\xb - \yb).
\end{equation}
As $(\pb, (\xb^i))$ is a competitive equilibrium, we have $x_j = s_j$ for all goods $j$ with $p_j > 0$. The feasibility of $(\qb, (\yb^i))$ implies $y_j \leq s_j = x_j$ for all $j$ with $p_j > 0$, so $\pb \cdot (\xb - \yb) \geq 0$. As $(\qb, (\yb^i))$ was an arbitrary feasible outcome, \eqref{eq:welfare-theorem} thus implies that $(\pb, (\xb^i))$ maximizes  social welfare among all feasible outcomes.

We now show that a constrained efficient outcome must also be a competitive equilibrium. Suppose that $(\qb, (\yb^i))$ maximizes social welfare among all feasible outcomes, so~$\sum_{i \in [m]} v^i(\yb^i) = \sum_{i \in [m]} v^i(\xb^i)$. By \eqref{eq:welfare-theorem}, we have
\begin{equation}
0  = \sum_{i \in [m]} v^i(\xb^i) - \sum_{i \in [m]} v^i(\yb^i) \geq \pb \cdot (\xb - \yb).
\end{equation}
If $(\qb, (\yb^i))$ is not a competitive equilibrium, we have $y_j \leq x_j$ for all goods $j$ for which $p_j > 0$, with strict inequality $y_j < x_j$ for at least one such good. Thus, $\pb \cdot (\xb - \yb) > 0$, a contradiction. 
\end{proof}

We note that \cref{proposition:welfare} continues to hold if the seller has non-zero costs. Costs for sold supply or, equivalently, utility for retained supply may be relevant in the application of financial asset exchanges. In most cases, leftover assets have value to the seller, or the issuance of financial instruments comes at a cost.

\section{The Revenue-Welfare Coincidence}
\label{section:revenue-welfare-coincidence}

We now assume that each buyer has a linear valuation. In this linear market setting, we establish that competitive equilibrium prices are unique and maximise not just constrained social welfare (\cref{proposition:welfare}) but also revenue. However, while \cref{proposition:welfare} tells us that constrained social welfare is maximised only at these prices, maximum revenue may also be attained at higher prices, at which the seller only sells a subset of her supply.

\begin{theorem}
\label{theorem:coincidence}
In the market with linear valuations, competitive equilibrium prices are unique. Moreover, all competitive equilibrium outcomes maximise revenue.
\end{theorem}

The linear valuation $v^i(\xb) = \vb^i \cdot \xb$ of each buyer $i$ is expressed by a \emph{valuation vector} $\vb^i$ representing his linear per-unit values $v^i_j \geq 0$ for each good~$j$. The per-unit value of money is $v^j_0 = 1$ for each buyer. Each buyer's utility function is then $u^{i}(\xb, \pb) = \vb^i \cdot \xb + \beta^i - \pb \cdot \xb$, and his \textit{demand correspondence} is $D^i(\pb) = \argmax_{\xb \in \Rn, \xb \cdot \pb \leq \beta^{i}} (\vb^i - \pb) \cdot \xb$. \cref{fig:arctic-bid} illustrates the demand that arises from linear valuations.

We assume that every good is valued positively by at least one buyer (otherwise we simply remove the good from the market). This guarantees that any feasible prices are strictly positive.

\begin{figure}[!tb]
\centering
\begin{subfigure}[t]{0.45\textwidth}
\centering
    \begin{tikzpicture}[scale=1]
    \pgfplotsset{scale=0.7}
    \begin{axis}[
    axis lines=middle,
    axis line style={->},
    x label style={anchor=west},
    y label style={anchor=south},
    xlabel={$p_A$}, ylabel={$p_B$},
    xtick = {5},
    ytick = {3},
    xmin=0, xmax=8,
    ymin=0, ymax=8,
    grid style=dashed,
    clip=false,
    ]
    \draw[margin] (0,0) -- (5, 3);
    \draw[margin] (5,3) -- (5,8);
    \draw[margin] (5,3) -- (8,3);
    \node[bundle] at (2.5,5) {$\left(\frac{\beta}{p_1},0\right)$};
    \node[bundle] at (7,1.5) {$\left (0,\frac{\beta}{p_2} \right)$};
    \node[bundle] at (7,5) {$\left(0,0 \right)$};
    \node[posbid] at (5,3) {};
    \end{axis}
\end{tikzpicture}
\end{subfigure}
\begin{subfigure}[t]{0.45\textwidth}
\centering
\begin{tikzpicture}
    \pgfplotsset{scale=0.7}
    \begin{axis}[
    axis lines=middle,
    axis line style={->},
    x label style={anchor=west},
    y label style={anchor=south},
    xlabel={$p_A$}, ylabel={$p_B$},
    xtick = {3,6},
    ytick = {3,6},
    xmin=0, xmax=8,
    ymin=0, ymax=8,
    grid style=dashed,
    clip=false,
    ]
    \node[posbid, label=left:$\bm{b}^1$] (b1) at (3,6) {};
    \node[posbid, label=below:$\bm{b}^2$] (b2) at (6,3) {};
    \coordinate (b1) at (3,6);
    \coordinate (b2) at (6,3);
    \draw[margin] (0,0) -- (b1);
    \draw[margin] (0,0) -- (b2);
    \draw[margin] (b1) -- (3,8);
    \draw[margin] (b1) -- (8,6);
    \draw[margin] (b2) -- (6,8);
    \draw[margin] (b2) -- (8,3);
    \node[bundle] at (7,7) {$\left(0,0 \right)$};
    \node[bundle] at (4.5,7) {$\left(\frac{\beta^2}{p_1},0\right)$};
    \node[bundle] at (7,4.5) {$\left (0,\frac{\beta^1}{p_2} \right)$};
    \node[bundle] at (4.5,4.5) {$\left (\frac{\beta^2}{p_1},\frac{\beta^1}{p_2} \right)$};
    \node[bundle] at (1.5,7) {$\left(\frac{\beta^1 + \beta^2}{p_1},0\right)$};
    \node[bundle] at (7,1.5) {$\left(0, \frac{\beta^1 + \beta^2}{p_2} \right)$};
\end{axis}
\end{tikzpicture}
\end{subfigure}

\caption{The demand of buyers with quasi-linear utilities, linear valuations and budget constraints, which divides price space into convex regions. In each region, we specify the bundle demanded, which depends on prices $\pb$.
\textbf{Left:} The demand of a single buyer with linear values $\vb = (5,3)$ and budget $\beta$ leads to three regions.
\textbf{Right:} The aggregate demand of two buyers, one with values $\vb^1 = (3,6)$ and budget $\beta^1$, and the other with values $\vb^2 = (6,3)$ and budget $\beta^2$.}
\label{fig:arctic-bid}
\end{figure}

In order to prove \cref{theorem:coincidence}, we first consider the feasible region, which consists of the set of all feasible prices (cf.~\cref{definition:feasible-region}). A key property of the feasible region is that it forms a lower semi-lattice. That is, there exists an elementwise-minimal price vector $\pb^*$ that is dominated by all other feasible prices. This is illustrated in \cref{fig:example-two}.
We develop these geometric insights in \cref{sec:elementwise-minimal-prices}. We then prove, in \cref{sec:elementwise-minimal-revenue}, that revenue is maximised, and the market is cleared, at these prices $\pb^*$.

For these proofs, we also make use of an alternative characterisation of the demand correspondence of a buyer with a linear valuation. Formally, we fix the price of the money at $p_0=1$.
At any given prices $\pb$, let $J^{i}(\pb) \coloneqq \argmax_{j \in [n]_0} \frac{v^{i}_j}{p_j}$ be the set of goods that maximise buyer $i$'s \emph{bang-per-buck}.
%
Note that $v^i_0 = p_0 = 1$ by definition, so $J^{i}(\pb)$ contains the \nullgood{} good $0$ if $\max_{j \in [n]_0}\frac{v^{i}_j}{p_j} = 1$, and we have $J^{i}(\pb) = \{0\}$ if $\max_{j \in [n]}\frac{v^{i}_j}{p_j} < 1$.

\Cref{lemma:argmax-demand} makes the observation that demanded bundles only contain quantities of goods that maximise a buyer's bang-per-buck; in other words, if $\xb$ is a bundle demanded by buyer $i$ at $\pb$, then $x_j > 0$ implies $j \in J^{i}(\pb)$. Moreover, any demanded bundle is the convex combination of the ``extremal'' bundles that arise when the entire budget is spent on a single demanded good in $J^{i}(\pb)$. In particular, $0 \in J^i(\pb)$ means that the buyer can maximise his utility by not spending his entire budget.

The convex hull of $S$, i.e., the set of all convex combinations of elements in $S$, is written as $\conv S$.

\begin{lemma}
\label{lemma:argmax-demand}
\label{lemma:bang-per-buck}
For any buyer $i$ with linear valuation $v^i$ and budget $\beta^i$, we have $D^i(\pb) = \conv \{\frac{\beta^i}{p_j} \eb^j \mid j \in J^i(\pb) \}$.
\end{lemma}
\begin{proof}
Recall that the buyer's demand is $D^i(\pb) = \argmax_{\xb \in \R^n_{+}, \xb \cdot \pb \leq \beta^i} (\vb^i-\pb) \cdot \xb$. The space $\{\xb \in \R^n_{+} \mid \xb \cdot \pb \leq \beta^i\}$ of all bundles that the buyer can afford at prices $\pb$ is a closed polyhedron spanned by vertices $\yb^0 \coloneqq \bm{0}$ and $\yb^j \coloneqq \frac{\beta^i}{p_j} \eb^j$ for each good $j \in [n]$. The fundamental theorem of linear algebra tells us that any bundle $\xb \in D^i(\pb)$ demanded at $\pb$ is the convex combination of the vertices $\yb^j$ that maximise $f(\xb) \coloneqq (\vb^i-\pb) \cdot \xb$. As $f(\yb^j) = (\vb^i-\pb) \cdot \frac{\beta^i}{p_j} \eb^j = \beta^i(\frac{v^i_j}{p_j} - 1)$ for each $j \in [n]_0$, we see that $\yb^j$ maximises $f(\yb)$ iff $j \in J^i(\pb)$.
\end{proof}

In \cref{sec:arctic}, we will also see that the Arctic Product-Mix Auction introduces a bidding language which starts from this definition of demand to characterise a more general class of preferences.

\subsection{Elementwise-Minimal Feasible Prices}
\label{sec:elementwise-minimal-prices}
Recall from \cref{definition:feasible-region} that prices $\pb$ are feasible if they can be extended with an allocation $(\xb^i)$ to a feasible outcome $(\pb, (\xb^i))$. We now show that the set of feasible prices form a lower semi-lattice. In particular, there exists a price vector $\pb^*$ that is elementwise smaller than all other feasible prices (so that $\pb^* \leq \pb$ for all feasible $\pb$).

\begin{proposition}
\label{proposition:elementwise-minimal-prices}
    The feasible region has a unique elementwise-minimal price vector $\pb^*$.
\end{proposition}

In order to develop the proof of \cref{proposition:elementwise-minimal-prices}, we first define the \emph{elementwise minimum} $\pb \wedge \qb$ of two prices $\pb$, $\qb$ by $(\pb \wedge \qb)_j = \min\{p_j, q_j\}$ for all goods $j \in [n]$.
The following lemma is central to our proof of \cref{proposition:elementwise-minimal-prices}.

\begin{lemma}
\label{lemma:elementwise-feasibility}
If $\pb$ and $\qb$ are feasible, then so is their elementwise minimum $\pb \wedge \qb$.
\end{lemma}

Fix feasible prices $\pb$ and $\qb$ with elementwise minimum $\rb \coloneqq \pb \wedge \qb$, and let $(\xb^i)$ and $(\yb^i)$ respectively denote allocations that extend $\pb$ and $\qb$ to feasible outcomes $(\pb, (\xb^i))$ and $(\qb, (\yb^i))$. In order to prove \cref{lemma:elementwise-feasibility}, we construct a third allocation $(\zb^i)$ and show that $(\rb, (\zb^i))$ is a feasible outcome. We first define the set of goods $A$ in which $\pb$ is strictly dominated by $\qb$, and its complement $B$, so $A = \{j \in [n] \mid p_j < q_j \}$ and $B = \{j \in [n] \mid p_j \geq q_j \}$. Then our allocation $(\zb^i)$ is given by
\begin{equation}
\label{eq:z-allocation}
    \zb^i =
    \begin{cases}
        \yb^i & \text{ if buyer $i$ demands some good $j \in B$ at $\rb$,} \\
        \xb^i & \text{ otherwise.}
    \end{cases}
\end{equation}

In order to prove that $(\rb, (\zb^i))$ is feasible, we first state a technical lemma that establishes the connection between a buyer's demand at $\pb$, $\qb$, and $\rb$.

\begin{lemma}
\label{lemma:A-set-demand}
\label{lemma:B-set-demand}
Suppose buyer $i$ demands good $j \in A$ at $\rb$. Then she also demands $j$ at $\pb$ and $J^{i}(\pb) \subseteq J^{i}(\rb)$. Similarly, suppose buyer $i$ demands $j \in B$ at $\rb$. Then she also demands $j$ at $\qb$, and $J^{i}(\qb) \subseteq J^{i}(\rb)$. Moreover, we have $J^{i}(\qb) \subseteq B$.
\end{lemma}
\begin{proof}
Fix a buyer $i$ who demands good $j \in A$ at $\rb$. As $p_j = r_j$, this implies $\frac{v^i_j}{p_j} = \frac{v^i_j}{r_j} \geq \frac{v^i_k}{r_k} \geq \frac{v^i_k}{p_k}$ for all goods $k \in [n]_0$. The first inequality holds due to \cref{lemma:bang-per-buck}, and the second inequality follows from the fact that $r_k \leq p_k$ for all $k \in [n]_0$. Hence, the buyer demands good $j$ at $\pb$.
For the second claim that $J^{i}(\pb) \subseteq J^{i}(\rb)$, fix a good $k \in J^{i}(\pb)$. Then we have $\frac{v^{i}_k}{r_k} \geq \frac{v^{i}_k}{p_k} \geq \frac{v^{i}_j}{p_j} = \frac{v^{i}_j}{r_j} \geq \frac{v^{i}_l}{r_l}$ for all goods $l \in [n]_0$. The first inequality holds due to $r_k \leq p_k$, and the second and third inequalities follow from the fact that the buyer $i$ demands $k$ at~$\pb$ and~$j$ at~$\rb$. Hence, if the buyer demands good $k$ at $\pb$, then they demand $k$ at $\rb$.
    
Now suppose that the buyer demands $j \in B$ at $\rb$. The proof of the first claim is identical to the case $j \in A$. We prove the last claim that $J^{i}(\qb) \subseteq B$. Suppose, for contradiction, that buyer $i$ demands a good $k \in A$ at $\qb$, and good $j \in B$ at $\rb$. This implies $\frac{v^{i}_k}{q_k} < \frac{v^{i}_k}{p_k} = \frac{v^{i}_k}{r_k} \leq \frac{v^{i}_j}{r_j} = \frac{v^{i}_j}{q_j}$, in contradiction to \cref{lemma:bang-per-buck} and the fact that $k$ is demanded at~$\qb$. 
\end{proof}

We can now prove \cref{lemma:elementwise-feasibility}.
\begin{proof}[Proof of \cref{lemma:elementwise-feasibility}]
Let $\pb$ and $\qb$ be two feasible prices and $\rb \coloneqq \pb \wedge \qb$ denote their elementwise minimum. As above, $(\xb^i)$ and $(\yb^i)$ are allocations that extend $\pb$ and $\qb$ to feasible outcomes, and $(\zb^i)$ is defined as in \eqref{eq:z-allocation}. We now prove that $(\rb, (\zb^i))$ is a feasible outcome, and so that it satisfies the two criteria in \cref{definition:feasible-outcome}. We can partition buyers into two sets: the set $\mathcal{B} \subseteq [m]$ of buyers that demand a good in $B$ at $\rb$, and the set $\mathcal{A} \coloneqq [m] \setminus \mathcal{B}$ of buyers that do not. Note that, by \cref{lemma:B-set-demand}, the buyers in $\mathcal{B}$ demand only goods in $B$ at $\qb$ and, by definition, the buyers in $\mathcal{A}$ demand only goods in $A$ at~$\pb$. We thus observe: in outcome $(\pb, (\xb^i))$, each buyer $i \in \mathcal{A}$ is only allocated quantities of goods in $A$ (so $\sum_{i \in \mathcal{A}} x^i_j > 0$ only if $j \in A$); similarly, in outcome $(\qb, (\yb^i))$, every buyer $i \in \mathcal{B}$ only receives quantities of goods in $B$ (so $\sum_{i \in \mathcal{B}} y^i_j > 0$ only if $j \in B$).

First we show that $(\zb^i)$ satisfies condition (i) of feasibility in \cref{definition:feasible-outcome}, i.e., that $\sum_{i \in [m]} z^i_j \leq s_j$ for all goods ${j \in [n]}$. Fix some good $j \in A$. Then, by definition of $(\zb^i)$, we have $\sum_{i \in [m]} z^i_j = \sum_{i \in \mathcal{A}} x^i_j + \sum_{i \in \mathcal{B}} y^i_j$. Recalling that $\sum_{i \in \mathcal{B}} y^i_j = 0$ (as $j \in A$) and that $(\pb, (\xb^i))$ is feasible, we get $ \sum_{i \in [m]} z^i_j \leq \sum_{i \in [m]} x^i_j \leq s_j$. This implies that~$(\rb, (\zb^i))$ does not over-allocate any goods $j \in A$. Analogously, we can show that $(\zb^i)$ does not over-allocate any goods $j \in B$ by recalling that $\sum_{i \in \mathcal{A}} x^i_j = 0$ for any good $j \in B$. As $A \cup B = [n]$, we have shown that $(\rb,\zb^i)$ satisfies condition (i) in \cref{definition:feasible-outcome}.

Next we argue that $(\rb, (\zb^i))$ satisfies condition (ii) of feasibility in \cref{definition:feasible-outcome}. Consider first a buyer $i \in \mathcal{A}$. By definition of $\mathcal{A}$, this buyer demands only goods in $A$ at $\rb$, and so by \cref{lemma:A-set-demand} we have $J^i(\pb) \subseteq J^i(\rb) \subseteq A$. Moreover, by \eqref{eq:z-allocation}, each buyer $i \in \mathcal{A}$ is allocated bundle $\zb^i = \xb^i$, and $x^i_j > 0$ implies $j \in A$. The prices of goods in $A$ are the same at $\pb$ and~$\rb$, by construction of $\rb$, so each buyer $i \in \mathcal{A}$ spends the same in outcome $(\rb, (\zb^i))$ as they do in outcome $(\pb, (\xb^i))$. As the latter outcome is feasible, we have $\zb^i \in D^i(\rb)$.
Similarly, as the buyers in $\mathcal{B}$ only demand goods in $B$, we apply the same argument to see that $\zb^i = \yb^i \in D^i(\rb)$ for every $i \in \mathcal{B}$.
\end{proof}

\begin{proof}[Proof of \cref{proposition:elementwise-minimal-prices}]
Suppose there exists no elementwise-minimal price vector. This means that, for all feasible $\pb$, there exists some feasible $\qb$ with $q_j < p_j$ for at least one good $j \in [n]$. Fix some feasible prices $\pb$ with the property that $\pb$ cannot be reduced any further in any direction without breaking feasibility. Such a point must exist, as the feasible region is closed and restricted to $\Rn_+$. By assumption, there exists a feasible price vector $\qb$ with $q_j < p_j$ for some $j \in [n]$.  Now consider $\rb = \pb \wedge \qb$. By \cref{lemma:elementwise-feasibility}, $\rb$ is feasible. But as $\rb \leq \pb$ with $r_j < p_j$, this contradicts our assumption that $\pb$ cannot be reduced~further.
\end{proof}

\subsection{Maximising Revenue and Welfare}
\label{sec:elementwise-minimal-revenue}
In \cref{sec:elementwise-minimal-prices}, we established that the set of feasible prices contains a unique elementwise-minimal price vector~$\pb^*$. We now show that revenue is maximised at these prices, and that the market is cleared at, and only at, $\pb^*$. Note that we do not assume $\pb^*$ to be the {only} prices at which revenue is maximised; indeed, there can be many revenue-maximising prices. However, a revenue-maximising outcome at $\pb^*$ clears the market, and is thus optimal for buyers among all revenue-maximising outcomes.

\begin{proposition}
\label{proposition:revenue}
The elementwise-minimal feasible prices $\pb^*$ maximise revenue.
\end{proposition}
\begin{proof}
We show that, for any feasible $\pb \leq \qb$, the maximum revenue obtainable at~$\pb$ is weakly greater than the revenue obtainable at $\qb$. It then follows immediately that revenue is maximised at $\pb^*$.

Let $(\xb^i)$ and $(\yb^i)$ be allocations that revenue-maximally extend $\pb$ and $\qb$ to feasible outcomes $(\pb, (\xb^i))$ and $(\qb, (\yb^i))$. Our goal is to determine an allocation $(\zb^i)$ so that $(\pb, (\zb^i))$ is a feasible outcome with a revenue that is weakly greater than the revenue of $(\qb, (\yb^i))$. As the revenue of $(\pb, (\xb^i))$ is weakly greater than the revenue of $(\pb, (\zb^i))$, the result then follows by transitivity.

If $\pb = \qb$, there is nothing to prove. Hence we assume that $S := \{j \in [n] \mid p_j < q_j \}$, the set of goods which are priced strictly lower at $\pb$ than at $\qb$, is non-empty. Fix a buyer $i \in [m]$. In order to define the new allocation $\zb^i$ to buyer $i$ at $\pb$, we distinguish between the two cases that $J^i (\pb)$ is, and is not, a subset of $S$.

\begin{description}
\item[Case 1:] Suppose buyer $i$ demands a subset of $S$ at $\pb$, so $J^i(\pb) \subseteq S$. In this case, we set $\zb^i \coloneqq \xb^i$. As $(\pb, (\xb^i))$ is a feasible outcome, we see that $\zb^i \in D^i(\pb)$. Moreover, as $0 \notin S$ by definition of $S$, the buyer spends his entire budget on $\xb^i$ at prices $\pb$.

\item[Case 2:] Suppose $J^{i}(\pb) \not \subseteq S$. We note that $J^{i}(\qb) \cap S = \emptyset$ and buyer $i$ demands all goods in $J^{i}(\qb)$ at $\pb$, so $J^{i}(\qb) \subseteq J^{i}(\pb)$. In this case, we set $\zb^i \coloneqq \yb^i$. As the buyer is only allocated goods not in $S$, and $p_j = q_j$ for all goods $j \in [n]_0 \setminus S$, it follows that the buyer spends the same at $\pb$ and $\qb$, and so $\zb^i \in D^i(\pb)$.
\end{description}
Note that, in both cases, the buyer is only allocated goods that they demand. To summarise, we define $(\zb^i)$ as
\[
    \zb^i \coloneqq  
    \begin{cases}
        \xb^i & \text{ if $J^i(\pb) \subseteq S$,} \\
        \yb^i & \text{ otherwise.}
    \end{cases}
\]

We now prove that $(\pb, (\zb^i))$ is a feasible outcome. We have already argued above that condition (ii) of \cref{definition:feasible-outcome} is satisfied. It remains to show that aggregate demand does not exceed supply $s_j$ for any goods $j \in [n]_0$. Note that the outcome $(\pb, (\zb^i))$ allocates a positive quantity of a good $j \in S$ to a buyer if and only if the buyer satisfies Case 1 above. Indeed, in this case we set $z^i_j = x^i_j$. Hence, for any $j \in S$, we have $\sum_{i \in [m]} z^i_j \leq \sum_{i \in [m]} x^i_j \leq s_j$ as $(\pb, (\xb^i))$ is feasible. Similarly, for any $j \not \in S$ the buyer will satisfy Case 2, and we get $\sum_{i \in [m]} z^i_j \leq \sum_{i \in [m]} y^i_j \leq s_j$ as $(\qb, (\yb^i))$ is feasible.

Finally, we see that $(\pb, (\zb^i))$ achieves weakly greater revenue than $(\qb, (\yb^i))$. Each buyer satisfying Case 1 spends his entire budget and thus contributes a weakly greater amount to overall revenue in outcome $(\pb, (\zb^i))$ than in outcome $(\qb, (\yb^i))$. A buyer satisfying Case 2 contributes the same amount in both outcomes.
\end{proof}

\begin{proposition}
\label{proposition:equilibrium}
The elementwise-minimal feasible prices $\pb^*$ uniquely clear the market.
\end{proposition}
\begin{proof}
We know that a competitive equilibrium $(\qb, (\yb^i))$ exists, as our market can be considered as an Arrow-Debreu exchange market \citep{Arrow-Debreu1954,Chen2007}. As this equilibrium is a feasible outcome, and $\pb^*$ is the elementwise-minimal feasible price vector, we have $\qb \geq \pb^*$.
Suppose, for the sake of contradiction, that $\qb \gneq \pb^*$. We let $(\xb^i)$ be an allocation that revenue-maximally extends $\pb^*$ to a feasible outcome. For convenience, let $\xb \coloneqq \sum_{i \in [m]} \xb^i$ be the aggregate demand of this outcome. As $(\qb, (\yb^i))$ clears the market, the respective revenues achieved by outcomes $(\qb, (\yb^i))$ and $(\pb^*, (\xb^i))$ are $\sum_{j \in [n]} q_j s_j$ and $\sum_{j \in [n]} p^*_j x_j$. As $\bm{0} \leq \xb \leq \ssb$ due to the feasibility of $(\pb^*, (\xb^i))$, and revenue is maximised at $\pb^*$ by \cref{proposition:revenue}, we arrive at the contradiction
\[
\sum_{j \in [n]} q_j s_j \leq \sum_{j \in [n]} p^*_j x_j \leq \sum_{j \in [n]} p^*_j s_j < \sum_{j \in [n]} q_j s_j.
\]
Here we use the fact that $\qb \gneq \pb^*$ for the strict inequality. The contradiction thus implies that $\pb^*$ are the unique market-clearing prices.
\end{proof}

\Cref{theorem:coincidence} now follows immediately from 
\cref{proposition:equilibrium,proposition:revenue}.

\section{Arctic Auctions}
\label{sec:arctic}
The market we consider can be interpreted as an important special case of the Arctic Product-Mix Auction (Arctic PMA). This auction was designed by \citet{Klemperer2018} for the government of Iceland, who wished to provide a mechanism for holders to exchange ``blocked'' offshore funds for alternative financial instruments. IMF staff has more recently proposed using a version of the Arctic PMA for sovereign debt restructuring, allowing creditors to exchange their claims for alternative debt instruments \citep{Baldwin2024language}.

When the sellers' costs are zero, the arctic PMA can be understood as a quasi-Fisher market, where money is priced at 1. As in our market, the seller in the Arctic PMA has a fixed supply $\vec{s} \in \Rn_+$ of $n$ divisible goods (financial assets) and wishes to find a feasible allocation of some subset of this supply among a finite set of buyers with the goal of maximising revenue. In the general Arctic PMA, the seller can additionally choose cost functions for their supply, while in our model, we assume the seller's costs to be zero.%
\footnote{The Arctic Product-Mix Auction is a variant of the original Product-Mix Auction developed for the Bank of England by Paul Klemperer \citep{Klemperer2008,klemperer2010product,Klemperer2018}. See also \citep{fichtl2022computing} for some discussion of the general case with non-zero costs.}

Each buyer is constrained to a budget that corresponds to the quantity of blocked offshore funds she holds. The buyers express their demand by submitting a collection of ``arctic bids'', each of which consists of an $n$-dimensional vector $\bid \in \Rn_+$ and a monetary budget $\beta(\bid)$. Each bid is associated with a demand correspondence, and the demands of a buyer's bids are then aggregated to yield the buyer's demand.

At a market price \emph{above the stated bid price}, an arctic bid \emph{rejects the good}. At market prices below the stated bid prices, the seller assigns the corresponding bid's budget to goods that yield the highest ``bang-per-buck'', i.e.,~the highest ratio of value to price $\frac{b_j}{p_j}$.
When multiple goods maximise the bid's bang-per-buck, the seller can arbitrarily divide the bid's budget between these goods. Moreover, if the maximal bang-per-buck is 1, then the seller can choose not to use some of the bid's budget (which we interpret as spending on the \nullgood{} good).
Hence, \cref{lemma:bang-per-buck} implies that each arctic bid $\bid$, interpreted in isolation, induces a quasi-linear demand with linear valuation $v(\xb) = \bid \cdot \xb$ and budget $\beta(\bid)$ as defined in \cref{sec:markets-preferences-objectives}. We denote the demand correspondence of each bid $\bid$ by~$D_{\bid}$.

The demand correspondence of a collection $\bids$ of bids is defined by the Minkowski sum of demands $D_{\bids}(\pb) \coloneqq \left\{ \sum_{\bid \in \bids} \xb^{\bid} \mid \xb^{\bid} \in D_{\bid}(\pb) \right\}$. Equivalently, $D_{\bids}$ can be understood as the aggregate demand of $|\bids|$ buyers with quasi-linear demand, linear valuations $\bid \in \bids$ and budgets $\beta(\bid)$. By submitting multiple arctic bids, buyers can express richer preferences. Nevertheless, it is straightforward to see from the definition of $D_\bids$ that, for the purposes of solving the auction (for welfare or for revenue), the seller can treat each bid independently, and we can assume without loss of generality that each buyer submits a single bid. Our results for quasi-Fisher markets on the coincidence of constrained efficiency and optimal revenue thus hold also for the arctic auction we describe above.

\section{Conclusion}\label{sec:conclusion}

In this article, we explore whether price-only mechanisms -- common in the digital economy and financial asset exchanges -- can achieve a form of efficiency that is attainable under budget constraints. We find that a unique set of prices is constrained utilitarian efficient, respecting buyers' budgets, and simultaneously revenue-optimal for the seller. This coincidence of revenue optimality and constrained efficiency makes our market compelling in theory and highly attractive to sellers, buyers, and market platforms in practice. Our results contribute to the policy debate on the regulation of digital monopolies. They suggest that a regulatory framework which requires digital advertising platforms to set the market-clearing price vector might unify the interests of the seller and the social planner (subject to the assumptions of our market setup), promoting competition and constrained social welfare.

Our approach is based on a geometric understanding of the structure underlying feasible market prices, which may be of independent interest. Future work includes addressing the open question of whether the revenue-welfare equivalence holds for other classes of preferences and markets with multiple buyers and sellers.

\bibliographystyle{te} 
\bibliography{refs}

\appendix
\section*{Omitted Proofs and Examples}
\label{app:sec:omitted-proofs}

\begin{proof}[Proof of \cref{prop:concave-v}]\phantomsection\label{proof:prop:concave-v}
Recall that at any price $p$, the buyer demands the bundle $x$ that maximises $v(x) - px$ subject to his budget constraint $px \leq \beta$. Given valuation $v(x)$, revenue is maximised at $(x,p)\in \argmax_{x,p} px$ such that $v'(x) = p$ and $x\leq s$. Thus, maximal revenue given $v$ and $s$ is $v'(x) x$ for some $x\leq s$. Social welfare is maximised at $(x,p)\in \argmax_{x,p} v(x) - px$ such that $x\leq s$, i.e.~revenue at the social optimum is $v'(s)s$.
We show that there exists $x$ and $s$ with $x < s$ and $v'(x) x > v'(s) s$. Recall that a differentiable function $f : \R \to \R$ is strongly concave if it satisfies $|f'(x) - f'(y)| \geq m \|x-y\|$ for all distinct $x, y \in \R$ ($f$ does not need to be twice differentiable).

First, note that there exists $\overline{x} < \infty$ such that $v'(x) < mx$ for all $x \geq \overline{x}$, due to strict concavity of $v$. Fix some supply $s \geq \overline{x}$ and let $\varepsilon = \frac{1}{2}(ms - v'(s))$. As $v$ is strongly concave, we also have  $v'(s - \varepsilon)  \geq v'(s) + m \varepsilon$ for any $\varepsilon$. Hence, $v'(s-\epsilon)(s-\epsilon) \geq (v'(s) + m\epsilon)(s-\epsilon) = v'(s) s + \epsilon(ms - v'(s) - m\epsilon) > v'(s) s$.
\end{proof}

\begin{proof}[Proof of \cref{prop:concave-v-2}]
If there exists supply $s \in X$ with $v'(s) < ms$, then the result follows analogously to the proof of \cref{prop:concave-v}.
We now prove the second part of the statement. The budget constraint is given by $px\leq \beta$. At the demanded quantity, $v'(x) = p$ holds. Thus, for all feasible $x$, it must hold that $v'(x) \leq \beta/x$, so~$x\leq\tilde{x}$. Now we demonstrate that $v'(\tilde{x}-\epsilon) \leq m(\tilde{x}-\epsilon)$ for some small $\epsilon$. Then the result follows from \cref{prop:concave-v-2}. For some $\delta > 0$, we have
\begin{align*}
    m (\tilde{x}- \epsilon) & \geq \left(\frac{v'(\tilde{x})}{\tilde{x}}+\delta\right) \left(\tilde{x} - \epsilon\right)
    = \frac{\left(v'(\tilde{x}) + \delta \tilde{x}\right)\left(\tilde{x}-\epsilon\right)}{\tilde{x}}
    \geq v'(\tilde{x} - \epsilon).
\end{align*}
The last inequality holds for $\epsilon \to 0$ and some $\delta > 0$.
\end{proof}

\begin{example}\label{example:concave-no-budget}
Consider an auction with a single good available in $s=3$ units that has a single buyer. The buyer has valuation $v:\R \to \R$ given by $v(x) = \frac{4}{\log 2} (1-2^{-x})$ and budget 2. Then revenue is not maximised at market-clearing prices.
Indeed, note that the utility of the buyer for quantity $q$ at price $p$ is $u(x,p) = v(x) - px$, so the buyer's demand $D(p)$ at $p$ is found by solving $v'(x) = p$, which yields $x = -\log_2(\frac{p}{4})$. At $p=0.5$, we have demand $x = 3$, so $p$ clears the market. Revenue at $p$ is $px = 1.5$. At price $q=1$, we have demand $y=2$, so $q$ does not clear the market, but revenue is $qy = 2$, which is greater. Revenue is maximised at $p = \frac{4}{e}$ with a demanded quantity of $\frac{1}{\log(2)}$ and a revenue of $\frac{4}{e \log(2)}$.
\end{example}
\end{document}